\documentclass[12pt]{article}

\usepackage{fullpage}
\usepackage[english]{babel}
\usepackage[latin1]{inputenc}

\usepackage{amssymb,amsmath,amsthm}
\usepackage{graphicx}
\usepackage{caption,subcaption}

\newtheorem{theorem}{Theorem}
\newtheorem{lemma}{Lemma}
\newtheorem{corollary}{Corollary}

\begin{document}

\title{\bf Optimal Cuts and Partitions in Tree Metrics in Polynomial Time}

\author{
  Marek Karpinski\thanks{Research partially supported by DFG grants and the Hausdorff grant~EXC59-1. Department of Computer Science, University of Bonn. Email:~{marek@cs.uni-bonn.de}}
  \and
  Andrzej Lingas\thanks{Research partially supported by VR grant 621-2008-4649. Department of Computer Science, Lund University. Email:~{andrzej.lingas@cs.lth.se}}
  \and
  Dzmitry Sledneu\thanks{Centre for Mathematical Sciences, Lund University. Email:~{dzmitry.sledneu@math.lu.se}.}
}

\date{}

\maketitle

\begin{abstract}
We present a polynomial time dynamic programming algorithm for optimal
partitions in the shortest path metric induced by a tree.  This resolves,
among other things, the exact complexity status of the optimal partition problems
in one dimensional geometric metric settings. Our method of solution could be also
of independent interest in other applications. We discuss also an
extension of our method to the class of metrics induced by the bounded treewidth graphs.
\end{abstract}


\section{Introduction}

The optimal partition problems
for unweighted and weighted graphs are classical NP-hard combinatorial
optimization problems.

The typical partition problems, like  MAX-CUT and MAX-BISECTION are well known to be APX-hard~\cite{Ho95}.
Also the existence  of a PTAS for MIN-BISECTION has been likely ruled out
in~\cite{K04}. The metric counterparts of these problems are to find the
corresponding optimal partitions of the complete graph on the input point set,
where edges are weighted by metric distances between their endpoints. The metric
MAX-CUT, MAX-BISECTION, MIN-BISECTION and other partitioning problems were all
proved to have \emph{polynomial time approximation
schemes}~(PTAS)~\cite{AFKK03,DK98,K02,FK98,FKK04,FKKV05,R10}. These problems are
known to be NP-hard in general metric setting (e.g., for $1-2$
metrics~\cite{DK98}). Their exact computational status for geometric (i.e., $L_p$) metrics and this
even for dimension one was widely open.

In this paper, in particular, we resolve the complexity status of these problems for just
dimension one by giving a polynomial time algorithm. Our solution, somewhat
surprisingly, involves certain new ideas for applying dynamic programming which
could be also of independent interest (see~\cite{KLS12} for a preliminary
version). In fact, our dynamic programming method works for the more general
case of tree metric spaces, where the underlying trees have nonnegative real edge
weights (see also the embeddability properties of arbitrary metrics in tree metrics~\cite{F03}).
Observe that the one dimensional case can be modeled by line graphs
with nonnegative real edge weights. We also give an evidence that our
polynomial time method can be extended to include analogous
partition problems in metric spaces induced by shortest paths in graphs of
constant treewidth.

\section{Preliminaries}

We  define our dynamic programming method in terms of generalized
subproblems on finite multisets of vertices in an undirected graph with
nonnegative real edge weights.

For a partition of a finite multiset $P$ of graph vertices into two multisets
$P_1$ and $P_2$, the \emph{value of the cut} is the sum over all pairs $(v,w)\in
P_1\times P_2$ of length (i.e., total weight) of a shortest path connecting $v$
with $w$. For example, see Fig. \ref{pic1a}.

The MAX-CUT problem for $P$ will be now to find a partition of $P$ into two
multisets that maximizes the value of the cut. If $|P|=n$ and the two multisets
are additionally required to be of cardinality $k$ and $n-k$, respectively, then
we obtain the $(k,n-k)$ MAX-PARTITION problem for $P$. In particular, if $n$ is
even and $k=n/2$ then we have the MAX-BISECTION problem. Next, if we replace the
requirement of maximization with that of minimization then we obtain the
$(k,n-k)$ MIN-PARTITION and MIN-BISECTION problems for $P$, respectively.

The aforementioned optimal partition problems can be generalized to include a
multiset of points in an arbitrary metric (in particular, geometric) space by
just replacing the weight of a shortest path connecting $v$ with $w$ with the
distance between $v$ and $w$. 

\begin{figure}[h!]
\centering
\begin{subfigure}{0.25\textwidth}
\includegraphics[width=\textwidth]{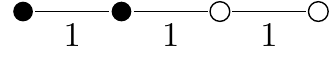}
\caption{}
\label{pic1a}
\end{subfigure}
\hspace{2cm}
\begin{subfigure}{0.25\textwidth}
\includegraphics[width=\textwidth]{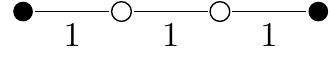}
\caption{}
\label{pic1b}
\end{subfigure}
\caption{The cut values of the first partition (a) and the second partition (b)
are respectively $8$ and $6.$ It follows that single geometric
cuts are not always sufficient to generate minimum bisections
on the real line.}
\end{figure}

Given a graph $G = (V, E)$, a \emph{tree decomposition} is a pair $(X, T)$,
where $X = \{X_1, \dots, X_l\}$ is a family of subsets (called bags) of $V$, and
$T$ is a tree whose nodes are the bags $X_i$, satisfying the following
properties~\cite[section 12.3]{D05}:
\begin{enumerate}
\item
The union of all bags $X_i$ equals $V$. 
\item For each $(v, w)\in E$, there is a bag $X_i$ that
  contains both $v$ and $w$. 
\item If $X_i$ and $X_j$ both contain a vertex $v$, then all nodes $X_k$ of the
  tree in the (unique) path between $X_i$ and $X_j$ contain $v$ as
  well. 
\end{enumerate}

The \emph{width} of a tree decomposition is the size of its largest set $X_i$
minus one. The treewidth $tw(G)$ of a graph $G$ is the minimum width among all
possible tree decompositions of $G$.

\section{The Algorithm}

Let $T$ be a tree with at most $n$ vertices and nonnegative real edge weights.
Next, let $P$ be a multiset of vertices of $T$ whose cardinality does not exceed
$n$.

Let us root $T$ at some vertex $r$. Next, for each vertex $v$ of $T$, let $T_v$
stands for the subtree of $T$ induced by all vertices of $T$ from which the path
to $r$ passes through $v$. We shall assume that $T_v$ is rooted at $v$. Finally,
let $P_v$ be the set of all elements in $P$ that are copies of vertices in
$T_v$, and let $P(v)$ stand for the set of all elements in $P$ that are copies
of the vertex $v$.

Consider a subtree $T_v$ of $T$ and a partition of $P$ into two multisets $A$
and $B$. The crucial observation is as follows.

In order to compute the total weight of shortest paths or their fragments
contained within $T_v$, it is sufficient to know the partition of $P_v$ into
$P_v\cap A$ and $P_v\cap B$, and just the number of elements in $P\setminus P_v$
that belong to $A$ or $B$.

Indeed, the specific placement of the elements of $(P\setminus P_v)\cap A$ or
$(P\setminus P_v)\cap B$ in the tree $T$ is not relevant when we consider solely
the maximal fragments of shortest paths connecting to these elements that are
contained within $T_v$.
 
The subproblem $MAXCUT(v,p,q,s,t)$ consists in finding a partition of $P_v$ into
two multisets $A$ and $B$ of cardinality $p$ and $q$, respectively, that
maximizes the sum of:

\begin{enumerate}
\item the total weight of shortest paths between
pairs of elements in $A\times B$;
\item $t$ times the total weight of the paths
between pairs of elements in $A\times \{ v\}$;
\item $s$ times the total weight of the paths
between pairs of elements in $B\times \{ v\}$.
\end{enumerate}

The parameters $s$ and $t$ are interpreted as the number
of elements in $P\setminus P_v$ that belong to the same
set as those in $A$ or $B,$ respectively, in the sought
two partition of $P.$

We shall denote the maximum possible value of the sum by $mc(v,p,q,s,t)$. Note
that the total number of subproblems is $O(n^3)$.

Assume first that $T$ is binary.

If $T_v$ is the singleton $(\{v\},\emptyset )$ then $MAXCUT(v,p,q,s,t)$ can be
trivially solved in $O(1)$ time. Otherwise, we can reduce the subproblem to
smaller ones by the recurrences given in the following lemmata.

\begin{lemma}\label{lem: twochild}
Suppose that $v$ has two children $v_1$ and $v_2$. The value $mc(v,p,q,s,t)$ is
equal to the maximum over partitions of $|P(v)|$ into the sum of natural numbers
$p(v),\ q(v)$, partitions of $|P_{v_1}|$ into the sum of natural numbers $p_1$
and $q_1$, and partitions of $|P_{v_2}|$ into the sum of natural numbers $q_1,\
q_2 $, where $p=p(v)+p_1+p_2$ and $q=q(v)+q_1+q_2$, of the total of:
\begin{enumerate}
\item $mc(v_1,p_1,q_1, s+p_2+p(v), t+q_2+q(v))$;
\item $mc(v_2,p_2,q_2, s+p_1+p(v), t+q_1+q(v))$;
\item $(p_1q_2+p_2q_1)(weight(v_1,v)+weight(v_2,v))$;
\item $p_1t\times weight(v_1,v)+ p_2t\times weight(v_2,v)$;
\item $q_1s \times weight(v_1,v)+q_2s\times weight(v_2,v)$.
\end{enumerate}
\end{lemma}

\begin{proof}
The first and second component values correspond to the total weights of
shortest paths or their fragments within $T_{v_i}$ connecting pairs of elements
in $P_{v_i}\times P_{v_i}$ or in $P_{v_i}\times (P\setminus P_{v_i})$, $i\in
\{1,2\},$ that belong to different sets in the sought two partition of $P$ (see
Fig.~\ref{pic2a}).

The three remaining ones correspond to the total weight of the edges $\{
v,v_1\}$ and $\{v,v_2\}$ on shortest paths between copies of vertices in $P$
belonging to different sets in the sought two partition. In particular, the
third component value corresponds to the total weight of the $\{ v ,v_1\}$ and
$\{v,v_2\}$ fragments of shortest paths connecting pairs $\{ a,b\}$ of elements
in $P$, where $a$ is a copy of a vertex in $T_{v_1}$ while $b$ is a copy of a
vertex in $T_{v_2}$ and $a,\ b$ belong to different sets of the two partition
(see Fig.~\ref{pic2b}).

Next, the fourth component value corresponds to the total weight of the $\{
v,v_1\}$ and $\{v,v_2\}$ fragments of shortest paths connecting pairs $\{ a,b\}$
of elements in $P$, where $a$ is a copy of a vertex in $T_{v_1}$ or $T_{v_2}$
belonging to the first set in the sought partition, while $b$ is a copy of a
vertex in $P\setminus P_v$ belonging to the second set in the partition. The
fifth component value can be specified symmetrically by switching the first set
with the second set in the two partition. For $p_1 t \times weight(v_1, v)$ term
see Fig.~\ref{pic2c}, the other terms are symmetrical.
\end{proof}

\begin{figure}[h!]
\begin{subfigure}{0.25\textwidth}
\includegraphics[width=\textwidth]{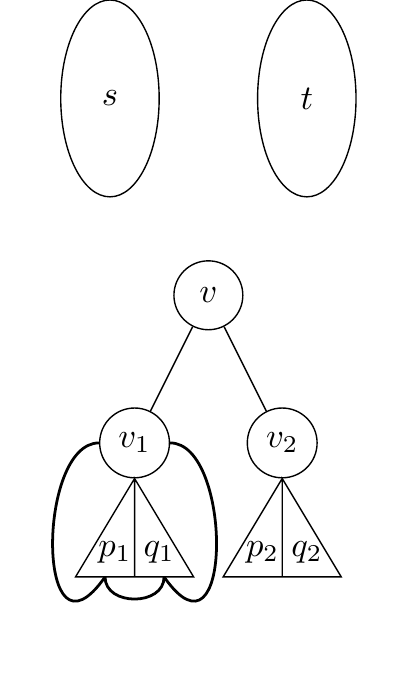}
\caption{}
\label{pic2a}
\end{subfigure}
\hfill
\begin{subfigure}{0.25\textwidth}
\includegraphics[width=\textwidth]{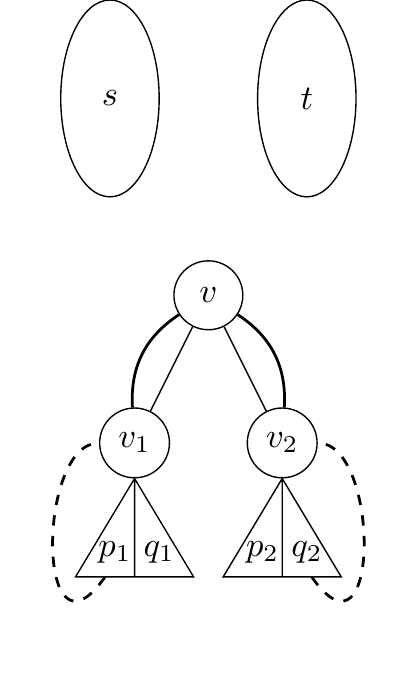}
\caption{}
\label{pic2b}
\end{subfigure}
\hfill
\begin{subfigure}{0.25\textwidth}
\includegraphics[width=\textwidth]{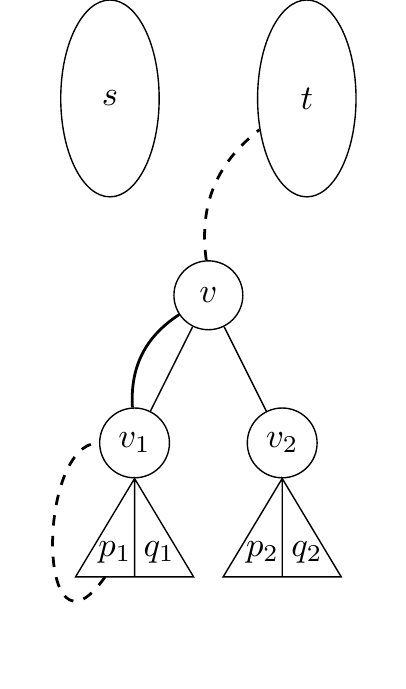}
\caption{}
\label{pic2c}
\end{subfigure}
\caption{Distinct cases in the proof of Lemma~\ref{lem: twochild}}
\end{figure}

\begin{lemma}\label{lem: onechild}
Suppose that $v$ has a single child $v_1$. The value $mc(v,p,q,s,t)$ is
equal to the maximum over partitions of $|P(v)|$ into the sum of natural numbers
$p(v),\ q(v)$, of the total of:
\begin{enumerate}
\item $mc(v_1,p-p(v),q-q(v), s+p(v), t+q(v))$;
\item $(p-p(v))t\times weight(v_1,v)$;
\item $(q-q(v))s\times weight(v_1,v)$.
\end{enumerate}
\end{lemma}

\begin{proof}
The first component value corresponds to the total weight of shortest paths or
their fragments within $T_{v_i}$ connecting pairs of elements in $P_{v_1}\times
P_{v_1}$ or in $P_{v_1}\times (P\setminus P_{v_1})$ that belong to different
sets in the sought two partition of $P$ (see Fig.~\ref{pic3a}). The second
component value corresponds to the total weight of the $\{ v,v_1\}$ fragments of
shortest paths connecting pairs $\{ a,b\}$ of elements in $P$, where $a$ is a
copy of a vertex in $T_{v_1}$ belonging to the first set in the sought
partition, while $b$ is a copy of a vertex in $P\setminus P_v$ belonging to the
second set in the partition (see Fig.~\ref{pic3b}). Finally, the third
component value can be specified symmetrically by switching the first set with
the second set in the two partition.
\end{proof}

\begin{figure}[h!]
\centering
\begin{subfigure}{0.25\textwidth}
\includegraphics[width=\textwidth]{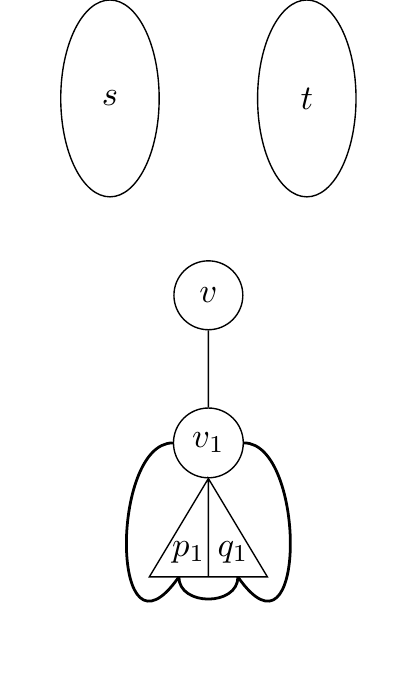}
\caption{}
\label{pic3a}
\end{subfigure}
\hspace{2cm}
\begin{subfigure}{0.25\textwidth}
\includegraphics[width=\textwidth]{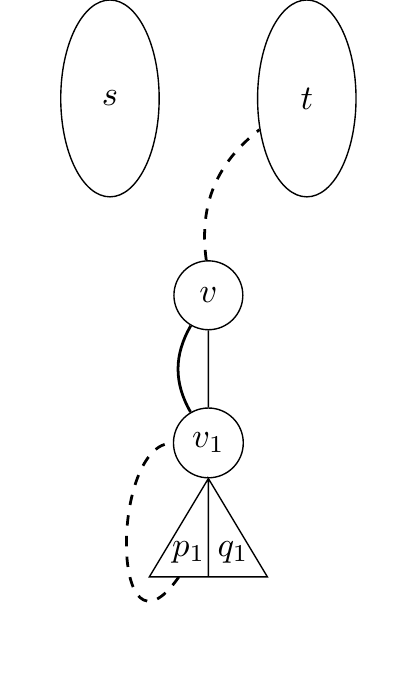}
\caption{}
\label{pic3b}
\end{subfigure}
\caption{Two cases in the proof of Lemma~\ref{lem: onechild}}
\end{figure}

In Lemma~\ref{lem: twochild}, after picking a partition of $|P(v)|$ and
$|P_{v_1}|$, the partition of $|P_{v_2}|$ is determined by the constraints
$p=p(v)+p_1+p_2$ and $q=q(v)+q_1+q_2$. It follows that the number of candidates
for partition sequences in Lemma~\ref{lem: twochild} does not exceed $O(n^2)$
while that number in Lemma~\ref{lem: onechild} is only $O(n)$. Furthermore, if
the input multiset $P$ is a set, i.e., $|P(v)|\le 1$ for all vertices $v$, then
the number of candidate partition sequences Lemmata~\ref{lem: twochild},
\ref{lem: onechild} is $O(n)$ and $2$, respectively. We conclude that we can
compute $mc(v,p,q,s,t)$ for all the corresponding subproblems
$MAXCUT(v,p,q,s,t)$ in a bottom up fashion in order of nondecreasing $|T_v|$ in
$O(n^{5})$ time in the general case and in $O(n^{4})$ time if $P$ is a set.

Now, it is sufficient to find the maximum among the values of the form
$mc(r,p,q,0,0)$ (recall that $r$ is the root of $T$), to obtain the maximum
value of a cut for $P$. If we additionally require $p=q$ then we obtain a
maximum value of bisection etc. By backtracking, we can construct maximum cut,
maximum bisection etc. The minimum variants of the partition problems can be
solved analogously.

If the rooted $T$ is not binary, we can easily transform it to a binary tree
$T'$ by adding totally at most $n-2$ dummy vertices between parents with more
than two children and their respective children. The dummy vertices between a
parent $v$ and its children form a binary subtree rooted at $v$ whose leaves are
the children. We set the weight of the edges incident to the children to the
corresponding weights of the edges between $v$ and their respective children in
$T$. All other edge weights are set to zero in the subtree. The considered
optimal cut and partitions problems for $P$ in $T$ are equivalent to those for
$P$ in $T'$.

\begin{theorem}
Let $T$ be a tree with at most $n$ vertices and nonnegative real edge weights,
and let $P$ be a multiset of vertices in $T$ with cardinality at most $n$. The
MAX-CUT, MAX-BISECTION, $(k,n-k)$ MAX-PARTITION, MIN-BISECTION and $(k,n-k)$
MIN-PARTITION problems for $P$ can be solved by dynamic programming in
$O(n^{5})$ time. If $P$ is a set then these problems can be solved in $O(n^{4})$
time.
\end{theorem}

The geometric optimal partition problems for a finite mulitset of points on the
real line can be equivalently formulated as the corresponding partition problems
in the shortest-path metric induced by the line graph whose vertices correspond
to the points in the multiset and whose edges correspond to the minimal
intervals between different points. The edges are weighted with the length of
intervals. Since in the case of a line graph, the reductions of a subproblem to
smaller ones rely solely on Lemma~\ref{lem: onechild}, the asymptotic time
complexity decreases by $n$.

\begin{corollary}\label{cor: max}
The geometric MAX-CUT problem on the real line as well as the geometric
MAX-BISECTION, $(k,n-k)$ MAX-PARTITION, MIN-BISECTION and $(k,n-k)$
MIN-PARTITION problems on the real line are solvable in $O(n^4)$ time. If the
input multiset of points is a set then these problems are solvable in in
$O(n^{3})$ time.
\end{corollary}

\section{Extensions to Graphs of Bounded Treewidth for Minimum Partition Problems}

A natural way of generalizing our method for tree metrics to those induced by
graphs of bounded treewidth is as follows. Construct a tree decomposition $T$ of
the input graph $G$. For a bag $b$ in the decomposition, let $T_b$ be the
subtree of $T$ rooted at $b$, and let $G_b$ stand for the subgraphs $G_b$ of $G$
induced by vertices contained in the bags of $T_b$. One could consider
analogously subproblems with $G_b$, the number of vertices in the first and the
second set of sought partition within $G_b $ and $G\setminus G_b$, respectively,
as the subproblem parameters.

Unfortunately, the aforementioned parameters do not seem sufficient to specify
such a generalized subproblem. Simply, if $b$ is not a singleton (as this is in
the case of trees) then it is not clear which vertex in $b$ on a shortest path
connecting a vertex in $G_b$ with a vertex in $G\setminus G_b$ belonging to a
different set in the sought partition, would be the last one on the path within
$G_b$.

One could try to tackle this problem by introducing $2|b|$ new parameters
specifying for each vertex in $b$ the number of shortest paths connecting
vertices in $G_b$ with those in $G\setminus G_b$ belonging to a different set in
the sought partition, where $v$ is the last vertex in $G_b$ on the shortest paths
before entering $G\setminus G_b$. (A single shortest path for each such pair
would be counted.) Then, one could generalize the recurrences given in
Lemmata~\ref{lem: twochild}, \ref{lem: onechild} but just in case of the minimum
partition problems.

The reason is that taking minimum would naturally force shortest path
connections in the solutions to the subproblems. On the contrary, taking maximum
would rather force longest path connections in these solutions. Thus, in the
latter case, the generalized method would not yield a correct solution to the
posed maximum partition problem.

Let us take a closer look at the generalized recurrences for minimum partition
problems. We may assume w.l.o.g that $T$ is binary~\cite{B96}. Let $b_1$ and
$b_2$ be the child bags of the bag $b$ in the tree decomposition $T$ of $G$. For
each of the aforementioned $2|b|$ parameters of a subproblem associated with
$G_b$, we need to consider possible partitions into at most $|b_1|+|b_2|$
components corresponding to the analogous parameters for the subproblems
associated with $G_{b_1}$ and $G_{b_2}$, respectively. Thus, if the width of the
tree decomposition $T$ is $d$, such a generalized recurrence would involve
$(n^{O(d)})^{O(d)}=n^{O(d^2)}$ optimal solutions to smaller subproblems. It
follows that the total time taken by the generalization to include minimum
partitions in metric spaces induced by shortest paths in $G$ would be
$n^{O(d^2)}$. The cost of the preprocessing, i.e., the construction of the
decomposition tree of $G$ of minimal width is marginal here~\cite{B96}.

\begin{theorem}
Let $G$ be a graph of treewidth $d$ with at most $n$ vertices and nonnegative
real edge weights, and let $P$ be a multiset of vertices in $G$ with cardinality
at most $n$. The MIN-BISECTION and $(k,n-k)$ MIN-PARTITION problems for $P$ can
be solved by dynamic programming in $n^{O(d^2)}$ time.
\end{theorem}

\section{Final Remarks}

Our dynamic programming method for optimal partitions in tree metrics yields in
particular polynomial time solutions to these problems on the real line. It
remains an open problem whether optimal partitions in geometric metrics of
dimension greater than $1$ admit polynomial time methods  or they turn out to
be inherently hard. At stake is the exact computational status of other
geometric problems for which our knowledge is very limited at the moment.
The exact computation complexity status of the maximum partition problems for shortest-path
metrics in bounded treewidth graphs is also a challenging open problem.

\section{Acknowledgments}

We thank Uri Feige, Ravi Kannan and Christos Levcopoulos for a number of
interesting discussions on the subject of this paper.

\bibliographystyle{abbrv}
\bibliography{bistreex}

\end{document}